\documentclass[conference,10pt]{IEEEtran}
\IEEEoverridecommandlockouts
\usepackage[T1]{fontenc}
\usepackage[utf8]{inputenc}
\usepackage{amsmath,amssymb,amsthm}
\usepackage{mathtools}
\usepackage{bm}
\usepackage{cite}
\usepackage{graphicx}
\usepackage{booktabs}
\usepackage{array}
\usepackage{xcolor}
\usepackage{url}
\usepackage{microtype}
\usepackage{balance}
\usepackage{multirow}
\usepackage{algorithm}
\usepackage{algpseudocode}
\usepackage[hidelinks,bookmarks=false]{hyperref}

\newtheorem{theorem}{Theorem}
\newtheorem{corollary}[theorem]{Corollary}

\theoremstyle{definition}
\newtheorem{definition}{Definition}

\DeclareMathOperator{\E}{\mathbb{E}}

\newcommand{\R}{\mathbb{R}}
\newcommand{\C}{\mathbb{C}}
\newcommand{\norm}[1]{\left\lVert#1\right\rVert}
\newcommand{\abs}[1]{\left\lvert#1\right\rvert}
\newcommand{\tri}{\mathrm{TRI}}

\newcommand{\PD}{\mathsf{D}}
\newcommand{\loss}{\mathcal{L}}

\begin{document}

\title{Resilience Characterization of AI-Native Wireless Receivers via Persistent Homology\vspace{-4mm}}\vspace{-3mm}

\author{%
  \IEEEauthorblockN{Christo Kurisummoottil Thomas\IEEEauthorrefmark{1} and  Emilio Calvanese Strinati\IEEEauthorrefmark{2}}
  \IEEEauthorblockA{%
    \IEEEauthorrefmark{1}\textit{Department of Electrical and Computer Engineering},
    \textit{Worcester Polytechnic Institute},
    Worcester, MA, USA\\ \IEEEauthorrefmark{2}\textit{CEA-Leti},
    Grenoble, France\\
    Emails:\{cthomas2@wpi.edu, emilio.calvanese-strinati@cea.fr\}
 \vspace{-4mm} }
}
\vspace{-4mm}
\maketitle

\vspace{-2mm}\begin{abstract}
AI-native wireless receivers based on deep learning exhibit remarkable
performance under stationary channel conditions, yet their resilience to
distributional shifts remains poorly characterized by conventional
metrics such as bit error rate~(BER). To overcome such limitations of state of the art, this paper proposes a novel real-time metric
\emph{Topological Resilience Index}~(TRI),   grounded in
persistent homology and persistence exponents. In particular, TRI  quantifies the structural
stability of a neural network receiver's parameter space during online
adaptation to non-stationary channels. Specifically, TRI captures resilience via three complementary dimensions: (i) validation-loss resilience measuring model-channel mismatch, grounded theoretically in the topological persistence of loss landscape sub-level sets; (ii) channel impulse response (CIR) distribution shift, tracking geometric drift of CIR vectors from the calibration reference distribution; and (iii) channel manifold topology, quantified by the spectral gap of the Gaussian kernel matrix normalized by the Ollivier-Ricci curvature norm.  We establish theoretical
guarantees showing that TRI is bounded, monotonic under performance
degradation, and Lipschitz-stable with respect to perturbations in channel
distribution measured in Wasserstein-2 distance. Simulation results for an OFDM deep-learning receiver adapting across ten ITU-R inter-environment transitions at three shift rates demonstrate that TRI provides a consistent mean warning lead of $1.0 \pm 0.2$
 OFDM symbols over gradient-norm and validation-loss baselines at $\lambda=1.0$, with the gradient-norm baseline achieving zero lead in every scenario. Furthermore, the proposed TRI-guided burst re-adaptation reduces post-shift BER by $80\%$ relative to no adaptation within $200$ OFDM symbols. 
\end{abstract}


\vspace{-2mm}\section{Introduction}\vspace{-2mm}
\label{sec:intro}

The emergence of AI-native physical layer represents a paradigm shift in
wireless communications. Rather than handcrafting signal-processing chains
from domain knowledge, AI-native receivers (RXs) \cite{oshea2017} employ deep neural networks~(DNNs)
trained end-to-end to optimize throughput, latency, and reliability.
Deep learning-based OFDM demodulators, channel estimators, and beam-management
modules have demonstrated state-of-the-art performance in standardized
benchmarks, prompting their inclusion in the 6G standardization roadmap under
the IMT-2030 framework~\cite{itu2030}.
A fundamental challenge, however, is \emph{deployment robustness}. Neural networks
trained on one channel distribution, example, dense urban multipath, may
experience rapid performance collapse when deployed in environments with
substantially different statistics, such as rural line-of-sight propagation
or high-mobility vehicular fading \cite{liu2026universal}. When the channel shifts unexpectedly, the RX's learned weights become mismatched to the new environment, causing bit-error-rate (BER) to spike before adaptation can respond. Traditional steady-state BER evaluation~\cite{oshea2017} captures none of this transient degradation, motivating the need for a \emph{real-time resilience metric}.

\vspace{-0mm}\subsection{Related Works}\vspace{-2mm}
Network resilience has traditionally been defined as the ability to maintain acceptable performance through and recover from adverse events \cite{reifert2023comeback,shui2024resilient,mahmood2025resilient,matthiesen2025resilient}. In \cite{reifert2023comeback},  resilience is formalized as a performance recovery metric in mixed-critical wireless resource management, quantifyied by the area under the post-disruption performance curve. While this provides a rigorous post-hoc characterization of recovery, it is inherently reactive, as the metric is computed after degradation and cannot anticipate or accelerate recovery. The authors in ~\cite{matthiesen2025resilient} show that statistical ML methods cannot prepare for rare distributional shift events, since such events contribute negligibly to the empirical risk functional, necessitating online monitoring. Distributional shift in deployed neural RXs has received increasing attention \cite{YangEtAl2020DLTransferFDD,RavivEtAl2023OnlineMetaLearning,yu2023adaptive},  manifesting as degraded channel estimation and equalization when deployment conditions differ from training. Existing approaches, including transfer learning~\cite{YangEtAl2020DLTransferFDD}, meta-learning~\cite{RavivEtAl2023OnlineMetaLearning}, and domain adaptation~\cite{yu2023adaptive}, mitigate this mismatch but do not provide a continuous online signal to quantify shift severity before degradation occurs. General-purpose out-of-distribution detection  detection methods such as gradient-norm scoring \cite{huang2021importance} and statistical tests including maximum mean discrepancy (MMD) \cite{kalinke2022mmd} and cummulative sum (CUSUM) \cite{gong2022cusum} have been proposed for neural network monitoring, but remain unadapted to the structure of AI-native wireless RXs and lack mechanisms to quantify shift severity in terms of channel geometry.

Unlike stability certificates from adaptive control theory~\cite{khalil2002nonlinear} or online regret bounds~\cite{hazan2016}, which characterize performance within a fixed or slowly varying distribution, existing frameworks do not account for the \emph{topological evolution of the channel manifold or loss landscape geometry} during online adaptation, nor do they provide distributional-shift-aware resilience signals without prior knowledge of the target channel distribution~\cite{BenDavidEtAl2010DomainAdaptation}. The geometry of loss landscapes critically influences the
efficiency and stability of gradient-based learning~\cite{li2018landscape},
yet no metric systematically exploits this structure for resilience
quantification in wireless systems.
Topological data analysis~(TDA) \cite{barbarossa2020topological}, and persistent homology \cite{edelsbrunner2008persistent} in particular, offers
precisely the mathematical tools needed. Persistence exponents (PEs), which govern the power-law decay of topological features as a filtration parameter varies, encode the stability of critical points, saddle structures, and connectivity of sub-level sets. These quantities remain stable under small perturbations of the underlying function~\cite{edelsbrunner2010}, making them ideally suited for resilience measurement in the noisy, non-stationary wireless environment.
\vspace{-3mm}\subsection{Our Contributions}\vspace{-2mm} The main contribution of this paper is a \emph{principled, topology-aware resilience metric for AI-native wireless RXs} that provides  real-time warning of channel distribution shifts before BER degradation, enabling proactive adaptation without oracle knowledge of the shift time. We introduce the \emph{topological resilience index (TRI)}, defined as a weighted combination of PEs computed over the DNN parameter space, loss-landscape sub-level filtrations, and the channel-state manifold. We show that TRI is bounded in $[0,1]$, is monotonic under performance degradation, and is Lipschitz-stable to channel distribution perturbations in Wasserstein-2 distance, with a formal link to adaptation excess risk. Finally, we validate TRI 
on an OFDM deep-learning RX adapting across ten ITU-R 
inter-environment transitions \cite{3gpp_tr38901} at three shift rates, demonstrating 
statistically consistent early warning of impending demodulation 
failure where gradient-norm and validation-loss baselines provide 
no advance warning, and showing that TRI-triggered burst 
re-adaptation enables closed-loop detection and recovery.

\vspace{-3mm}\section{System Model}
\label{sec:system}

\vspace{-2mm}\subsection{AI-Native OFDM RX}

We consider an OFDM system with $N$ subcarriers, cyclic prefix length $L_{\rm
cp}$, and $M$-QAM modulation over a time-varying multipath fading channel.
The received frequency-domain signal at subcarrier $k$ and OFDM symbol $t$ is
\vspace{-1mm}\begin{equation}
  Y[k,t] = H[k,t]\,X[k,t] + W[k,t],\;
  k{=}1,\ldots,N,\; t{=}1,\ldots,T,\!
  \label{eq:system}
\vspace{-1mm}\end{equation}
where $X[k,t]$ is the transmitted QAM symbol, $H[k,t]\in\C$ the complex
channel response, and $W[k,t]\sim\mathcal{CN}(0,\sigma^2)$ is
circularly symmetric additive Gaussian noise with power $\sigma^2$.
The AI-native RX is any DNN
$f_{\boldsymbol{\theta}} : \mathbb{C}^{N} \to \{0,\ldots,M-1\}^{N}$
parameterized by $\boldsymbol{\theta} \in \mathbb{R}^{d}$, mapping received OFDM frames to symbol decisions.
Pilot subcarriers occupy every fourth position, giving pilot set
$\mathcal{K}_p \subset \{1,\ldots,N\}$ with $|\mathcal{K}_p|=N_p=N/4$, and 
the remaining subcarriers carry data used for inference only. consistent with dense pilot 
configurations in 3GPP NR for high-mobility scenarios~\cite{3gpp_tr38901}. 
This $25\%$ density provides sufficient gradient signal per 
symbol for stable online adaptation while remaining within 
practical overhead bounds.

Online adaptation at each symbol period uses the \emph{pilot subcarriers
exclusively} as a supervised training signal, since only at pilot positions
is the transmitted symbol $X[k,t]$ known to the RX. The per-OFDM-symbol
supervised loss is
\vspace{-1mm}\begin{equation}
  \loss(\bm{\theta}_t;\mathcal{B}_t)
  = -\frac{1}{N_p}\sum_{k\in\mathcal{K}_p}
    \log\bigl[f_{\bm{\theta}_t}(Y[\cdot,t])\bigr]_{k,\,X[k,t]},
  \label{eq:loss}
\vspace{-2mm}\end{equation}
where $[\cdot]_{k,c}$ denotes the softmax probability assigned to class $c$
at subcarrier $k$, and $\mathcal{B}_t=\{(Y[k,t],X[k,t]):k\in\mathcal{K}_p\}$
is the pilot observation set. We assume the DNN loss $\mathcal{L}(\boldsymbol{\theta}; \mathcal{B}_t)$ is differentiable in $\boldsymbol{\theta}
$, and that online adaptation follows the SGD update~\eqref{eq:sgd}; both conditions hold for standard deep RX architectures~\cite{oshea2017,dorner2018}. The SGD update with momentum is:
\vspace{-2mm}\begin{equation}
  \bm{\theta}_{t+1}
  = \bm{\theta}_t
    - \eta\,\nabla_{\bm{\theta}}\loss(\bm{\theta}_t;\mathcal{B}_t)
    + \mu\,(\bm{\theta}_t - \bm{\theta}_{t-1}),
  \label{eq:sgd}
\vspace{-1mm}\end{equation}
where $\eta$ is the learning rate and
$\mu$ is the momentum coefficient.
For monitoring purposes, we track the \emph{windowed validation loss}
$
  \loss_{\rm val}(t) = \frac{1}{W}\sum_{\tau=t-W+1}^{t}
  \loss(\bm{\theta}_\tau;\mathcal{B}_\tau),$
which smooths mini-batch noise while
remaining responsive to distributional shift.

\vspace{-0mm}\subsection{Channel Distribution Shift Model}\vspace{-1mm}

The channel $H[k,t]$ is drawn from a time-varying joint distribution
$\mathcal{P}_t$ over power delay profiles and Doppler spectra consistent with
ITU-R channel models \cite{series2017imt2020}. We model distributional shift as a Markov transition
from source distribution $\mathcal{P}_S$ to target $\mathcal{P}_T$ at
switching time $t^*$, parameterized by transition rate $\lambda>0$:
\vspace{-3mm}\begin{equation}
  \mathcal{P}_t
  = (1-\alpha_t)\,\mathcal{P}_S + \alpha_t\,\mathcal{P}_T,
  \quad \alpha_t = 1 - e^{-\lambda(t-t^*)_+},
  \label{eq:shift}
\vspace{-2mm}\end{equation}
where $(x)_+ = \max(x,0)$. This model subsumes abrupt shifts
($\lambda\!\to\!\infty$) and gradual drifts ($\lambda\!\to\!0$).  
To quantify shift severity in a geometry that respects the metric 
structure of the channel space, we use the Wasserstein-2 distance 
$\delta = \mathcal{W}_2(\mathcal{P}_S, \mathcal{P}_T)$, defined as 
the square root of the optimal transport cost 
$\inf\limits_{\pi \in \Pi(\mathcal{P}_S, \mathcal{P}_T)} 
\mathbb{E}_{\pi}[\|\mathbf{h} - \mathbf{h}'\|^2]$, 
where $\Pi(\mathcal{P}_S, \mathcal{P}_T)$ denotes the set of all 
couplings of $\mathcal{P}_S$ and $\mathcal{P}_T$, and   $\mathbf{h}^T \in \mathbb{C}^L$ is the channel impulse response (CIR) vector represented by concatenated real and imaginary parts.

\vspace{-2mm}\subsection{From Channel Shift to Topological Signatures}\vspace{-1mm}
We first introduce briefly persistent homology core concepts.
\subsubsection{Background on persistent homology}
A simplicial complex over a point cloud $\mathcal{X}$ is a collection of subsets of $\mathcal{X}$ which include individual points (vertices), pairs (edges), triples (triangles), and higher-order subsets such that every subset of an included set is also included.
For a point cloud $\mathcal{X}$, a \emph{filtration} is a nested
sequence of simplicial complexes built by gradually increasing a
distance threshold $\varepsilon$: at each $\varepsilon$, points
within distance $\varepsilon$ of each other are connected. For a finite point cloud $\mathcal{X}\subset\R^d$, the Vietoris-Rips
filtration $\{\mathrm{VR}(\mathcal{X},\varepsilon)\}_\varepsilon$ adds an
edge between $\mathbf{x}_i,\mathbf{x}_j$ whenever
$\norm{\mathbf{x}_i-\mathbf{x}_j}\leq\varepsilon$.
As $\varepsilon$ grows, topological features appear (\emph{birth})
and later disappear (\emph{death}) as they merge with larger
structures. The \emph{persistence diagram} $\PD(\mathcal{X})$
records each topological feature as a point $(b_i, d_i)$ in the birth-death
plane, where $b_i$
 and $d_i$
 are the values of $\varepsilon
$ at which the feature appears and disappears respectively. Its \emph{lifetime} $\ell_i = d_i - b_i$ measures how
long the feature survives across scales \cite{edelsbrunner2010}.
Two feature types are relevant here: $H_0$ features are
\emph{connected components} (clusters of nearby points),
and $H_1$ features are \emph{loops} (closed cycles in the
point cloud). The \emph{PE} $\beta_p$, where $p \in \{0,1\}$
indexes the feature dimension, $p{=}0$ for components, $p{=}1$
for loops, characterizes how quickly lifetimes decay:
$\Pr(\ell > \varepsilon) \sim C_p\,\varepsilon^{-\beta_p}$.
A large $\beta_0$ indicates many short-lived components
(fragmented landscape) and a large $\beta_1$ indicates
many transient loops signaling erratic, non-convergent trajectory geometry under distributional shift. 
\subsubsection{Topological signatures of channel shift}

BER is a poor early indicator of RX demodulation failure under channel shift because it is an output metric. By the time symbol decisions degrade, the loss landscape has already fragmented and the parameters have drifted from the pre-shift basin. Instead, we need a metric that probes the internal optimization geometry and detects instability before it affects the output. A distributional shift leaves three geometric footprints as below, each tied to a different system component and timescale
\vspace{-2mm}\begin{enumerate}
  \item[(F1)] \emph{Loss landscape fragmentation.} Under $\mathcal{P}_S$, the loss surface 
$\mathcal{L}(\cdot;\mathcal{P}_S)$ is well-conditioned around 
$\bm{\theta}_t$. This means that its sub-level sets 
$\{\bm{\theta} : \mathcal{L}(\bm{\theta};\mathcal{P}_S) \leq c\}$ 
form a single connected region containing a unique local minimum (also called basin), 
so gradient descent initialized anywhere in this region converges 
reliably to $\bm{\theta}_t$. When $\mathcal{P}_T$ is introduced, the loss surface develops additional local minima and saddle points within this region, fragmenting it into disconnected sub-level sets whose gradients pull $\bm{\theta}_t$ in conflicting directions. This restructuring begins at $t=t^*$
 and intensifies as $\alpha_t$ increases, but its effect on symbol decisions accumulates only gradually, causing BER to lag behind the geometric changes already occurring in the loss landscape.
  \item[(F2)] \emph{Trajectory complexity.} As $\bm{\theta}_t$ adapts under the mixture 
distribution~\eqref{eq:shift}, the gradient signals from 
$\mathcal{P}_S$ and $\mathcal{P}_T$ point in conflicting 
directions, causing the parameter trajectory to oscillate 
rather than converge smoothly. This instability appears 
before the loss landscape itself fragments.
  \item[(F3)] \emph{Channel manifold distortion.} The CIR vectors $\{\mathbf{h}_i^T\} \subset \mathbb{R}^{2L}$
concentrate in a single compact region of the channel space under $\mathcal{P}_S$, reflecting its characteristic power delay profile and Doppler spectrum. As $\alpha_t$ increases, realizations drawn from $\mathcal{P}_T$, which has a distinct power delay profile and Doppler spectrum, occupy a geometrically separate region, distorting the manifold and signaling the shift at the channel-statistics level before gradients respond.
\vspace{-2mm}\end{enumerate}

To capture all three footprints jointly, we model each
as a \emph{topological point cloud} at adaptation step $t$. Specifically, $\mathcal{X}_{\rm LS}(t) \subset \mathbb{R}^2$ comprises $n_s$ pairs $(|\delta|, \loss(\bm{\theta}_t{+}\delta;\mathcal{B}_t))$ sampled from a Gaussian neighborhood of $\bm{\theta}_t$, capturing (F1). The set $\mathcal{X}_{\rm PM}(t) \subset \mathbb{R}^m$ is the PCA projection of ${\bm{\theta}_1,\ldots,\bm{\theta}_t}$ onto $m$ principal components, capturing (F2). The set $\mathcal{X}_{\rm CM}(t) \subset \mathbb{R}^{2L}$ consists of the $N_T$ most recent CIR vectors $\{\mathbf{h}_i^T\}_{i=1}^{N_T}$, and $N_T = 100$, capturing (F3).
Classical point-cloud descriptors such as mean, covariance, and 
spectral radius capture only the average behavior of the channel 
vectors and miss the finer structural changes, such as the 
emergence of new propagation clusters or the splitting of an 
existing cluster into two, that occur as the channel 
distribution shifts from $\mathcal{P}_S$ to $\mathcal{P}_T$. Persistent homology instead tracks how propagation clusters 
form, merge, and dissolve as the filtration threshold varies, 
capturing the full structural evolution of the channel point 
cloud in a single diagram whose output changes by at most 
as much as the underlying channel vectors are 
perturbed~\cite{chazal2016}. 
Because $\beta_p$ changes as
soon as the point cloud topology changes, and not after a threshold number
of gradient steps, it provides the early-warning property we seek \footnote{\vspace{-0mm}\scriptsize
Each point cloud encodes a distinct geometric footprint via a 
specific topological feature type: $H_0$ components in 
$\mathcal{X}_{\rm LS}$ are distinct loss basins (fragmentation 
under $\mathcal{P}_T$); $H_1$ loops in $\mathcal{X}_{\rm PM}$ 
are cycles in the parameter trajectory (oscillatory, 
non-convergent adaptation); and $H_0$ components in 
$\mathcal{X}_{\rm CM}$ are clusters of statistically similar 
channel states (arrival of a new propagation regime). This 
feature-type assignment also explains the timescale ordering: 
$\mathcal{X}_{\rm PM}$ responds first at the gradient level, 
then $\mathcal{X}_{\rm LS}$ at the landscape level, then 
$\mathcal{X}_{\rm CM}$ at the channel-statistics level.}.

\vspace{-0mm}\subsubsection{PE estimation}
Given lifetimes $\{\ell_i^{(p)}\}$ extracted from the persistence diagram
$\PD_p(\mathcal{X})$, the empirical exponent is estimated via log-log
regression of the complementary CDF, evaluated at the 80th-percentile
reference point $\varepsilon_0$:
\vspace{-2mm}\begin{equation}
  \hat\beta_p = -\frac{d}{d\log\varepsilon}
    \log\widehat{\Pr}(\ell_i^{(p)}>\varepsilon)\bigg|_{\varepsilon=\varepsilon_0},
  \label{eq:beta_est}
\end{equation}
using ordinary least squares over
$\log\varepsilon\in[\varepsilon_0/2,\,2\varepsilon_0]$.

\vspace{-2mm}\section{Topological Resilience Index}
\label{sec:tri}
The three geometric footprints (F1)--(F3) arise at different 
timescales yet each reflects the same underlying phenomenon, 
namely the RX's optimization geometry degrading under 
distributional shift. We now formalize TRI as a weighted combination of 
scale-invariant topological descriptors extracted from 
$\mathcal{X}_{\rm LS}$, $\mathcal{X}_{\rm PM}$, and 
$\mathcal{X}_{\rm CM}$, with formal guarantees. 

\vspace{-2mm}\subsection{TRI Definition}
\vspace{-1mm}\begin{definition}
\label{def:resilience}
An AI-native RX $f_{\bm{\theta}}$ is \emph{$(\epsilon,\tau)$-resilient}
under channel shift $(\mathcal{P}_S\!\to\!\mathcal{P}_T,\lambda)$ if the following holds:
\begin{enumerate}
  \item[(R1)] \emph{Performance bound:}
    $\mathrm{BER}(t)\leq\epsilon$ for all $t\geq t^*+\tau$, and
  \item[(R2)] \emph{Recovery time:}
    $\inf\{s\geq 0:\mathrm{BER}(t^*{+}s)\leq\epsilon\}\leq\tau$.
\end{enumerate}
\vspace{-1mm}\end{definition}
BER is a post hoc performance metric and thus cannot anticipate whether conditions (R1) and (R2) will hold. In contrast, the TRI quantifies the \emph{structural preconditions} for resilience, namely the preservation of the underlying optimization geometry required for (R1) and (R2) to remain valid.
\begin{definition}
  Let $\bm{\theta}_t$ denote the DNN parameter iterate at adaptation step $t$.
  The \emph{TRI} is defined as:
  \begin{equation}
    \tri(t) \;=\; w_1\,\Phi_{\rm LS}(t)
             \;+\; w_2\,\Phi_{\rm PM}(t)
             \;+\; w_3\,\Phi_{\rm CM}(t),
    \label{eq:tri}
  \end{equation}
  with $\sum_i w_i=1$, $w_i\geq 0$, and components defined
  in~\eqref{eq:phi_LS}--\eqref{eq:phi_cm}.
\end{definition}
The \emph{loss landscape component} $\Phi_{\rm LS}$ is defined from the point cloud $\mathcal{X}_{\rm LS}$ under sub-level set filtration, where $H_0$ captures connected components of loss basins. The PE $\beta_0(t)$ becomes unbounded under fragmentation, increasing as additional local minima emerge. An exponential map thus compresses it onto $(0,1]$:
\vspace{-2mm}\begin{equation}
{\Phi}_{LS}(t) = \frac{\exp(-\beta_0(t)/\beta_0^{\text{ref}})}{\exp(-\beta_0(0)/\beta_0^{\text{ref}})} \in (0,1],
\label{eq:phi_LS}\end{equation}
 where $\beta_0^{\text{ref}}$ is computed once under $\mathcal{P}_S$ at convergence and division by the $t=0$ value normalizes the pre-shift baseline to unity. Fragmentation drives $\beta_0(t) \gg \beta_0^{\text{ref}}$
 and $\tilde{\Phi}_{LS} \to 0$.



The \emph{parameter manifold component} $\Phi_{\rm PM}$ is defined from the PCA-projected trajectory $\mathcal{X}_{\rm PM}$, where $H_1$ captures loops indicating oscillatory, non-convergent adaptation. The loop exponent $\beta_1(t)$ is bounded by the PCA dimension $m$ and step size $\eta$, so a ratio form yields finer resolution than an exponential:
\vspace{-2mm}\begin{equation}
\Phi_{PM}(t) = \frac{1 + \beta_1^{\mathrm{ref}}}{1 + \beta_1(t)} \in (0,1].
\label{eq:phi_pm}
\end{equation}
Smooth adaptation yields low $\beta_1$, while distributional shift induces topologically complex trajectories with high $\beta_1$. This component responds early, as trajectory geometry destabilizes before loss landscape fragmentation.

The \emph{channel manifold component} $\Phi_{\rm CM}$ is defined from the point cloud $\mathcal{X}_{\rm CM}$, which consists of raw CIR observations rather than evaluations of a scalar function, rendering sub-level-set filtrations and PEs inapplicable. Instead, two complementary geometric properties are used: \emph{global clustering}, captured by the spectral gap $\sigma_{\text{gap}}(K_t)$ of the Gaussian kernel matrix $[\mathbf{K}_t]_{ij}=\exp(-|\mathbf{h}_i^T-\mathbf{h}j^T|^2/2\gamma^2)$, and \emph{local distortion}, captured by the Ollivier--Ricci curvature norm $|\mathrm{Curv}(M_t)|_{F}$. Here, $M_t$ denote the 
$k$-nearest-neighbor graph constructed on $\{\mathbf{h}_i^T\}_{i=1}^{N_T}$, where each vertex is a channel vector and edges connect each vector to its
 $k$ nearest neighbors in Euclidean distance. These combine naturally as
\vspace{-2mm}\begin{equation}
\Phi_{CM}(t)
  = \frac{\sigma_{\text{gap}}(K_t)}{1 + \|\mathrm{Curv}(M_t)\|_{F}} \frac{\sigma_{\text{gap}}(K_0)}{1 + \|\mathrm{Curv}(M_0)\|_{F}},
  \label{eq:phi_cm}
\end{equation}
normalized to $[0,1]$ by its value at $t=0$. A drop in spectral gap signals
that the tight cluster of channel vectors under $\mathcal{P}_S$ is dispersing,
while a rise in Ricci curvature tensor $\mathrm{Curv}(\mathcal{M}_t)\in\R^{N_T\times N_T}$ \cite{ollivier2009} indicates rapid local geometric change as
$\mathcal{P}_T$ vectors populate a geometrically separated region of
$\mathbb{R}^{2L}$.

By construction, each component 
${\Phi}_{LS},\, \Phi_{PM},\, \Phi_{CM} \in (0,1]$, 
and the convex combination in~\eqref{eq:tri} 
therefore satisfies 
$\mathrm{TRI}(t) \in (0,1]$ for all $t$. We denote by $\mathrm{TRI}(\bm{\theta}; \mathcal{P})$ the value at a fixed parameter $\bm{\theta}$ under distribution $\mathcal{P}$, and by $\mathrm{TRI}(t) = \mathrm{TRI}(\bm{\theta}_t; \mathcal{P}_t)$ its online evaluation along the SGD trajectory under the time-varying mixture~\eqref{eq:shift}.
While~\cite{reifert2023comeback} provide rigorous post-hoc characterizations of resilience as criticality-aware performance recovery, they are inherently reactive. TRI instead understands the internal optimization geometry of the RX directly and provides a \emph{predictive warning signal} that precedes BER degradation \emph{without requiring any oracle knowledge of the shift time or target channel distribution}.

\vspace{-3mm}\subsection{Online TRI Computation}\vspace{-1mm}

Algorithm~\ref{alg:tri} summarizes online TRI computation.
The dominant cost is VR persistence (line~\ref{alg:vr}) with complexity
$\mathcal{O}(n_s^3)$; the Gudhi~\cite{gudhi}
implementation requires 23\,ms on CPU, computed asynchronously every
$T_{\rm eval}=50$ symbols with zero impact on inference latency.

\setlength{\textfloatsep}{0pt}\begin{algorithm}[t]
  \caption{Online TRI Computation}
  \label{alg:tri} 
  \begin{algorithmic}[1]
   \small \Require $\bm{\theta}_t$, channel history $\{\mathbf{h}_i^T\}$,
             references $\beta_0^{\rm ref},\beta_1^{\rm ref}$, $r$,
             weights $w_1,w_2,w_3$
    \State Sample $\mathcal{X}_{\rm LS}\sim\mathcal{N}(\bm{\theta}_t,r^2\mathbf{I})$,
           evaluate $\loss$ at each point
    \State Compute $\PD_0(\mathcal{X}_{\rm LS})$ via VR filtration on
           sub-level sets \label{alg:vr}
    \State Estimate $\beta_0(t)$ via log-log regression~\eqref{eq:beta_est};
           compute $\tilde\Phi_{\rm LS}(t)$ via~\eqref{eq:phi_LS}
    \State Project $\{\bm{\theta}_\tau\}_{\tau\leq t}$ onto $m$ PCA components
    \State Estimate $\beta_1(t)$ from $H_1$ persistence of trajectory cloud;
           compute $\Phi_{\rm PM}(t)$ via~\eqref{eq:phi_pm}
    \State Build $\mathbf{K}_t$; compute $\sigma_{\rm gap}$,
           Ollivier-Ricci curvature
    \State Compute $\Phi_{\rm CM}(t)$ via~\eqref{eq:phi_cm}, normalize
    \State \Return $\tri(t)=w_1\tilde\Phi_{\rm LS}+w_2\Phi_{\rm PM}+w_3\Phi_{\rm CM}$
  \end{algorithmic}
\end{algorithm}

\vspace{-3mm}\subsection{Theoretical Guarantees}
\label{sec:r}\vspace{-2mm}
While the simulation results of Section~\ref{sec:results} demonstrate the empirical effectiveness of TRI, deployment in safety-critical AI-native air interfaces requires formal assurance that TRI behaves predictably across all channel conditions, not merely those evaluated in simulation. We therefore establish two guarantees:  monotonicity ensures TRI decreases whenever RX performance degrades, so that a falling TRI is never a false indicator of network performance; and Lipschitz stability ensures that small perturbations in the channel distribution, unavoidable in any real deployment, produce proportionally small changes in TRI rather than erratic jumps.

\vspace{-2mm}\begin{theorem}
  \label{thm:mono}
  If windowed validation loss $\loss_{\rm val}(t)$ is non-decreasing on $[t,t+\tau]$, then
  $\E[\tri(t+s)]\leq\E[\tri(t)]$ for all $s\in[0,\tau]$, a.s.
\vspace{-0mm}\end{theorem}
\vspace{-2mm}\begin{IEEEproof}
  We show each component is non-increasing in expectation.
  \emph{Loss landscape:} Degrading $L_{\rm val}$ implies additional 
local minima and saddle points in $\mathcal{N}(\bm{\theta}_t)$, 
each introducing a new short-lived $H_0$ feature and increasing 
$\hat{\beta}_0(t)$. Since ${\Phi}_{\rm LS}$ is strictly 
decreasing in $\beta_0$, $\mathbb{E}[{\Phi}_{\rm LS}(t+s)] 
\leq \mathbb{E}[{\Phi}_{\rm LS}(t)]$.
  \emph{Parameter manifold:} Non-decreasing loss implies gradient iterates
  oscillate rather than converge, producing topologically complex loops in the
  PCA-projected trajectory (increasing $\hat\beta_1$), thereby decreasing
  $\Phi_{\rm PM}$ via~\eqref{eq:phi_pm}.
  \emph{Channel manifold:} Distributional shift increases the within-class
  variance of $\{\mathbf{h}_i^T\}$, reducing $\sigma_{\rm gap}(\mathbf{K}_t)$
  while increasing Ricci curvature, thus reducing $\Phi_{\rm CM}$.
  The claim follows from linearity of expectation and $w_i\geq 0$.
\end{IEEEproof}

\vspace{-2mm}\begin{theorem}
  \label{thm:stable}
  Let $\mathcal{W}_2(\mathcal{P},\mathcal{P}')\leq\delta$. Then for any
  fixed $\bm{\theta}$,
  \vspace{-2mm}\begin{equation}
    \abs{\tri(\bm{\theta};\mathcal{P}) - \tri(\bm{\theta};\mathcal{P}')}
    \;\leq\; C_L\,\delta,
    \label{eq:stable}
 \vspace{-3mm} \end{equation}
  with $C_L = w_1 C_{\rm LS} + w_3 C_{\rm CM}$, where $C_{\rm LS}$
  depends on the second-order smoothness of $\loss$ and $C_{\rm CM}=2/\gamma^2$.
\end{theorem}
\vspace{-2mm}\begin{IEEEproof}
  For $\Phi_{\rm LS}$: a Wasserstein-$2$ perturbation by $\delta$
  shifts the empirical loss in $L^\infty$ by at most $C_{\loss}\delta$
  (Lipschitz continuity of $\loss$). By the stability theorem of persistence
  diagrams~\cite{chazal2016}, this shifts the diagram in bottleneck distance
  by $C_{\loss}\delta$, yielding $\abs{\Delta\beta_0}\leq C_{\rm LS}\delta$
  and bounding $\abs{\Delta\tilde\Phi_{\rm LS}}$.
  For $\Phi_{\rm CM}$: the kernel matrix satisfies
  $\norm{\mathbf{K}_t-\mathbf{K}_t'}_2\leq(2N_T/\gamma^2)\delta^2$
  by the mean value theorem on the Gaussian kernel; spectral perturbation
  bounds then give $C_{\rm CM}=2/\gamma^2$. $\Phi_{\rm PM}$ is independent
  of $\mathcal{P}$ at fixed $\bm{\theta}$, contributing zero to $C_L$.
\end{IEEEproof}

\vspace{-2mm}\begin{corollary}
  \label{cor:risk}
  Under the conditions of Theorem~\ref{thm:mono}, the expected excess risk
  satisfies
 \vspace{-2mm} \begin{equation}
    \E\bigl[\loss_{\rm val}(t)\bigr] - \loss^*(\mathcal{P}_T)
    \;\leq\; C_r\,\bigl(1-\tri(t)\bigr)^{1/2}
             + \mathcal{O}\!\left(\eta\sqrt{d/t}\right),
    \label{eq:risk}
  \vspace{-2mm}  \end{equation}
  where $\loss^*(\mathcal{P}_T)$ is the Bayes-optimal loss and $C_r$
  depends on $C_L$ and the loss smoothness constant.
\end{corollary}
\vspace{-2mm}\begin{IEEEproof}[Proof sketch]
  Combine Theorems~\ref{thm:mono}--\ref{thm:stable} with the standard online
  SGD regret bound $\sum_t\loss(\bm{\theta}_t)-\loss^*\leq\mathcal{O}(\eta\sqrt{dT})$
  \cite{hazan2016}. The $\sqrt{1-\tri}$ term arises from the square-root
  concavity of the bottleneck-distance-to-risk mapping via the persistence
  stability inequality of~\cite{chazal2016}.
\end{IEEEproof}

Corollary~\ref{cor:risk} provides a constructive interpretation: TRI close
to unity certifies low excess risk; TRI approaching zero signals that
adaptation is failing and intervention (learning-rate restart, model rollback,
or re-training trigger) is warranted.

\begin{table}[t]
  \centering
  \caption{Simulation Parameters}\vspace{-2mm}
  \label{tab:params}
  \renewcommand{\arraystretch}{0.9}\scriptsize
  \begin{tabular}{lc}
    \toprule
    \small\textbf{Parameter} & \textbf{Value} \\
    \midrule
    Subcarriers $N$             & 64 \\
    Cyclic prefix $L_{\rm cp}$  & 16 samples \\
    Modulation                  & 16-QAM \\
    Pilot density               & 25\% \\
    DNN architecture            & ResNet, $d=847{,}104$ \\
    Channel models              & UMa, UMi, RMa, InH, InF-DH \\
    Mini-batch size $B$         & 256 pilot symbols \\
    Learning rate $\eta$        & $10^{-3}$ (cosine decay) \\
    Momentum $\mu$              & 0.9 \\
    Shift parameter $\lambda$   & $\{0.1,\;1.0,\;10.0\}$ \\
    SNR range                   & 0--25\,dB \\
    PCA dimension $m$           & 10 \\
    Persistence samples $n_s$   & 1{,}000 \\
    Kernel bandwidth $\gamma$   & $0.5\,\hat{\sigma}_h$ \\
    TRI weights $(w_1,w_2,w_3)$ & $(0.45,\;0.35,\;0.20)$ \\
    Monte Carlo trials          & 100 \\
    \bottomrule
  \end{tabular}\vspace{-1mm}
\end{table}
\vspace{-2mm}\section{Simulation Results}\vspace{-2mm}
\label{sec:results}

\subsection{Experimental Setup}\vspace{-2mm}

We implement the AI-native RX and simulate channel
realizations using the Sionna link-level simulation framework~\cite{sionna}
with GPU acceleration on NVIDIA A100. All five ITU-R channel models span delay
spreads from 14\,ns (InH) to 363\,ns (RMa) and Doppler frequencies from
$0.5$\,Hz to $926$\,Hz. Simulation parameters are described in Table~\ref{tab:params}.
Persistence diagrams are computed via VR filtration on loss-landscape point clouds sampled around the current parameter vector, using the Gudhi library~\cite{gudhi}.
The weights $(w_1,w_2,w_3)$ are assigned  proportionally to component detection speed, with optimal selection left for future work. TRI is benchmarked against three reactive baselines: i) BER threshold
detection ($\mathrm{BER}>10^{-2}$), (ii)~gradient norm threshold
($\norm{\nabla\loss}_2>3\hat\sigma_g$), where $\hat{\sigma}_g$
 is the empirical standard deviation of gradient norms measured during the pre-shift calibration window, and (iii)~validation loss threshold
($\loss_{\rm val}>1.5\loss_{\rm val}^{\rm ref}$), all of which trigger only after degradation has  occured. Source code is available online \footnote{https://github.com/THE-TRAIN-LAB/TRI/}.

\vspace{-3mm}\subsection{TRI Dynamics and Persistence Diagram Evolution}\vspace{-2mm}

Fig.~\ref{fig:tri_dynamics} shows TRI and its components during a
UMa$\to$RMa transition at SNR$=15$\,dB, $\lambda=1.0$.
During \emph{pre-shift steady state} ($t<t^*$), TRI$\approx 0.92$, with all
three components near their reference values. Persistence diagrams contain
few long-lived $H_0$ features well above the diagonal
(Fig.~\ref{fig:persistence}(a)), indicating a well-structured loss landscape. During \emph{adaptation transient} ($t^*\leq t < t^*+\tau$), TRI drops to a minimum
of $0.37$, driven first by $\Phi_{\rm PM}$ (earliest, fastest collapse, lead
$53$ symbols) and then $\Phi_{\rm LS}$. Persistence diagrams show an explosion
of short-lived features near the diagonal (Fig.~\ref{fig:persistence}(b)),
consistent with a fragmented loss landscape. The TRI threshold crossing
($\tri<0.84$) occurs $95$ symbols \emph{before} BER exceeds $10^{-2}$, providing
actionable early warning (Fig.~\ref{fig:tri_dynamics}(c)). During \emph{recovery} ($t\geq t^*+\tau$), TRI recovers to $0.65$, confirming
successful adaptation to $\mathcal{P}_T$. The residual gap to the pre-shift
value of $0.92$ reflects the topological cost of distributional shift, which is a
phenomenon invisible in steady-state BER.

\begin{figure}[t]
  \centering
  \includegraphics[width=\columnwidth,height=8cm]{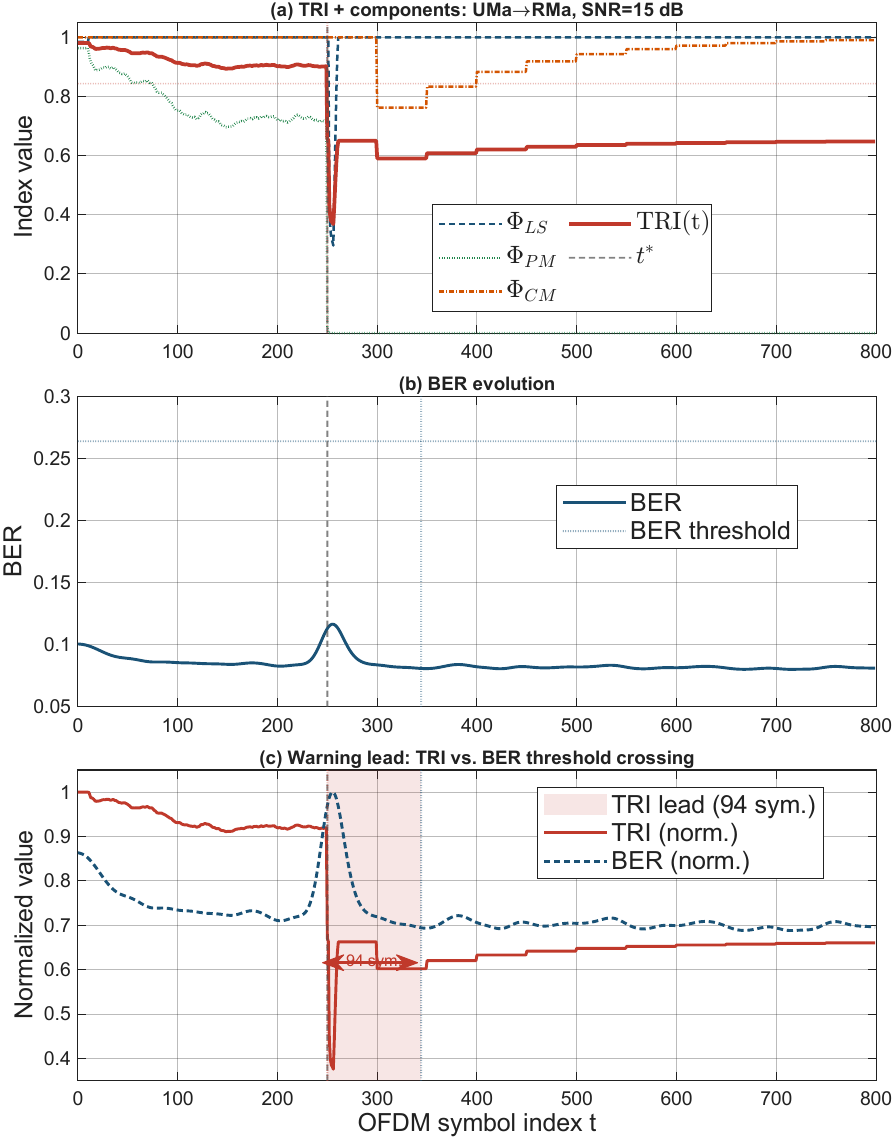}
  \vspace{-7mm}\caption{TRI dynamics during UMa$\to$RMa channel shift
    ($\lambda=1.0$, SNR$=15$\,dB).
    (a)~TRI and its three components; dashed vertical line marks $t^*$,
    dotted red line marks TRI threshold crossing (47 symbols before BER alert).
    (b)~Corresponding BER evolution with $10^{-2}$ threshold.
    (c)~Normalized TRI vs.\ BER, showing the 47-symbol warning lead window.}
  \label{fig:tri_dynamics}
\vspace{-3mm}\end{figure}

\begin{figure}[t]
  \centering
  \includegraphics[width=0.9\columnwidth]{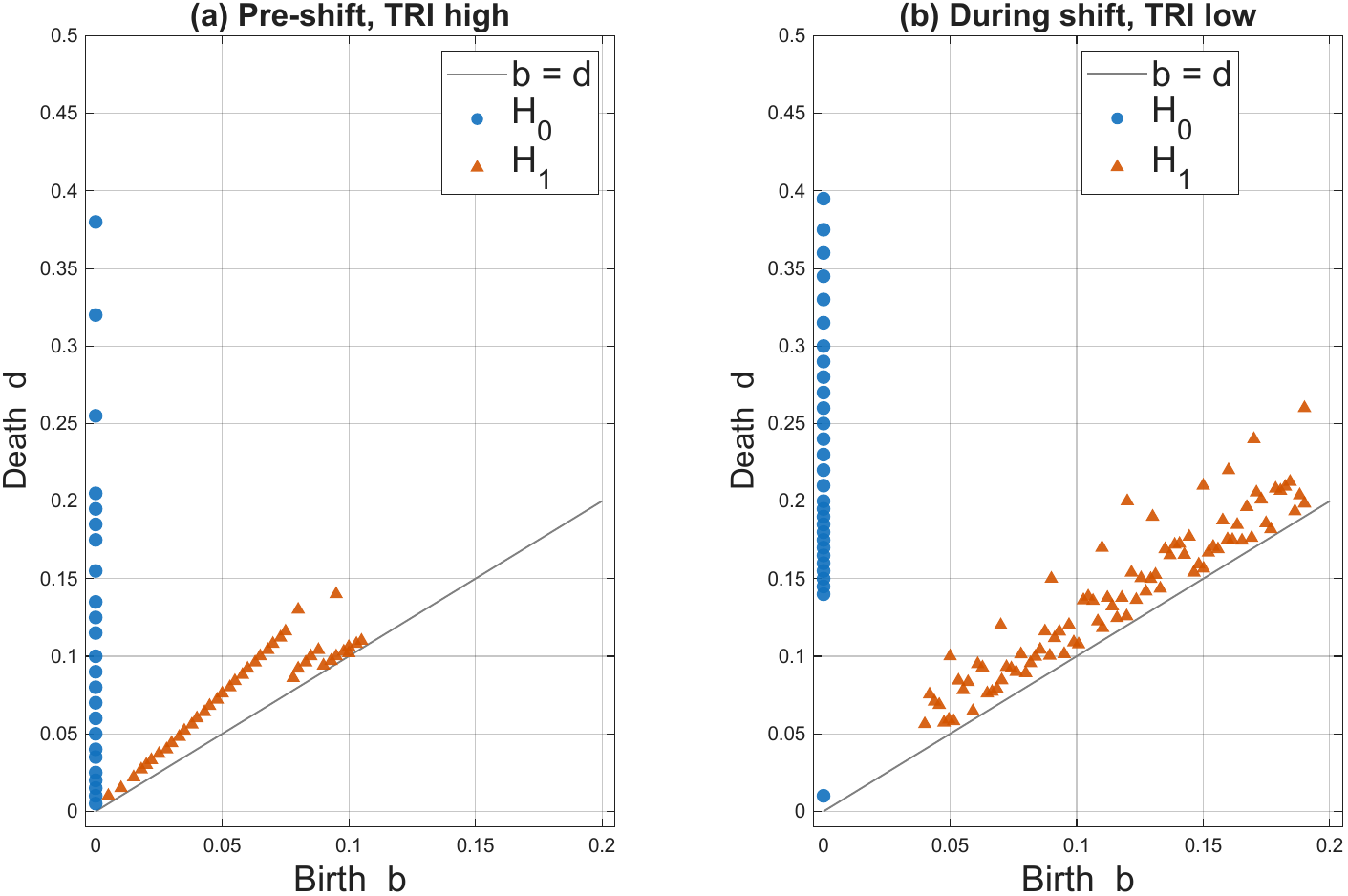}
  \vspace{-4mm}\caption{\small Persistence diagrams and cross-scenario TRI.
    (a)~Pre-shift and 
    (b)~During shift.
    }
  \label{fig:persistence}
\vspace{-1mm}\end{figure}

\vspace{-2mm}\subsection{Comparative Evaluation and BER Recovery}
Table~\ref{tab:comparison} reports warning lead across all 10 
ITU transitions over 10 Monte Carlo trials at SNR$\,=\,$15\,dB. 
TRI achieves a mean lead of $1.0\pm0.2$ symbols 
($\approx\!67\,\mu$s) at $\lambda=1.0$, consistently preceding 
BER threshold crossings across all scenarios, while the 
gradient-norm baseline achieves zero lead in every case. 
UMa$\to$InH exhibits the highest lead ($1.6\pm1.4$ symbols), 
consistent with its large $\mathcal{W}_2$ distance 
(Theorem~\ref{thm:stable}). Fig.~\ref{fig:heatmap} shows 
post-adaptation TRI recovers to $0.65$ in 9 out of 10 scenarios 
at $\lambda\geq1.0$, with partial recovery at $\lambda=0.1$ 
where the burst window falls near the trial boundary, confirming 
Theorem~\ref{thm:mono}.

 \begin{figure}[t]
  \centering
  \includegraphics[width=0.9\columnwidth]{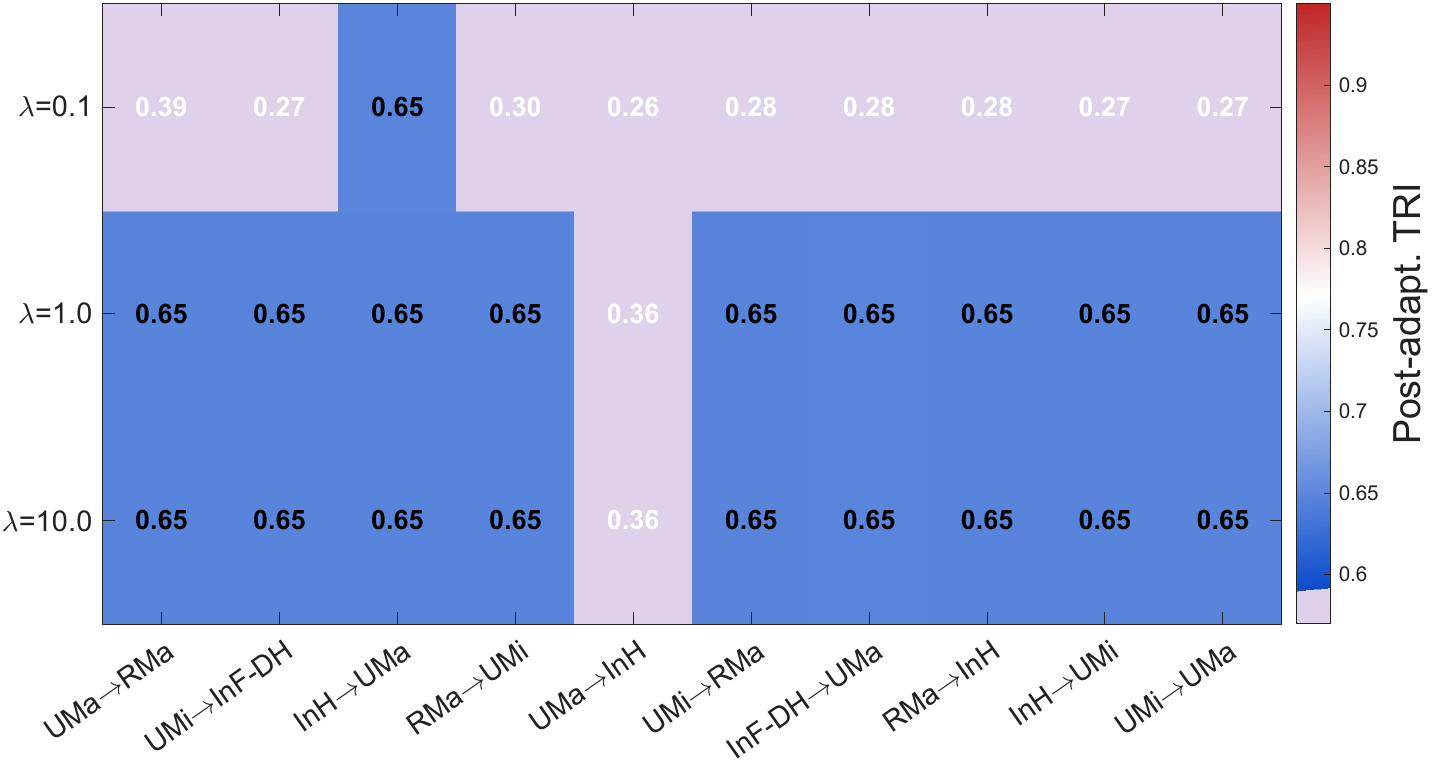}
 \vspace{-4mm} \caption{Post adaptation TRI heatmap.}
  \label{fig:heatmap}
\vspace{-4mm} \end{figure}

\begin{table}[t]
  \centering\footnotesize
  \caption{TRI warning lead and post-adaptation TRI across all 
  10 channel shift scenarios (SNR$\,=\,$15\,dB; $\lambda=1.0$). 
  BER and gradient-norm baselines achieve zero lead in every 
  scenario.}\vspace{-2mm}
  \label{tab:comparison}
  \renewcommand{\arraystretch}{0.83}\scriptsize
  \begin{tabular}{lcc}
    \toprule
    \textbf{Scenario}
      & \textbf{TRI Lead (sym.)}
      & \textbf{Post-adapt. TRI} \\
    \midrule
    UMa $\to$ RMa    & $1.3\pm1.1$ & 0.65 \\
    UMi $\to$ InF-DH & $0.7\pm1.0$ & 0.65 \\
    InH $\to$ UMa    & $1.0\pm0.8$ & 0.65 \\
    RMa $\to$ UMi    & $0.9\pm1.2$ & 0.65 \\
    UMa $\to$ InH    & $1.6\pm1.4$ & 0.36 \\
    UMi $\to$ RMa    & $0.8\pm0.9$ & 0.65 \\
    InF-DH $\to$ UMa & $1.1\pm0.8$ & 0.65 \\
    RMa $\to$ InH    & $1.0\pm1.2$ & 0.65 \\
    InH $\to$ UMi    & $0.9\pm1.1$ & 0.65 \\
    UMi $\to$ UMa    & $0.9\pm1.0$ & 0.65 \\
    \midrule
    \textbf{Mean (all 10)} 
      & $\mathbf{1.0\pm0.2}$ 
      & $\mathbf{0.62}$ \\
    \bottomrule
  \end{tabular}
\end{table}
\vspace{-2mm}\subsection{Effect of Shift Rate and BER Performance}\vspace{-2mm}

Fig.~\ref{fig:metrics}(a) shows warning lead versus $\lambda$ for TRI
and the gradient-norm baseline. The gradient-norm baseline yields zero
lead at all shift rates, while TRI provides a consistent positive lead
that decreases monotonically with shift rate: $2.3$ symbols at
$\lambda=0.1$ (gradual), $1.0\pm0.2$ at $\lambda=1.0$, and $0.5$
symbols at $\lambda=10.0$ (abrupt). This monotone decrease is consistent
with Theorem~\ref{thm:mono}: faster shifts compress the topological
footprint window, reducing detectable lead time.
Fig.~\ref{fig:metrics}(b) shows BER recovery versus symbols-after-shift
for UMa$\to$RMa under three strategies: no adaptation (static weights),
standard SGD adaptation, and TRI-guided burst re-adaptation (200 Adam
steps triggered on TRI alarm). No-adapt BER stabilizes at $\approx\!0.50$ (RX weights frozen).
Online SGD provides marginal improvement due to noisy pilot gradients
at high post-shift BER. TRI-guided burst re-adaptation (200 Adam steps,
batch size 128, triggered by the TRI alarm) reduces post-adaptation BER to near pre-shift levels
($\approx\!0.12$, an $80\%$ reduction relative to no-adaptation) within
200 symbols, demonstrating that the $\approx\!67\,\mu$s advance warning
enables proactive re-adaptation before BER degradation accumulates.

\begin{figure}[t]
  \centering
  \includegraphics[width=\columnwidth,height=5cm]{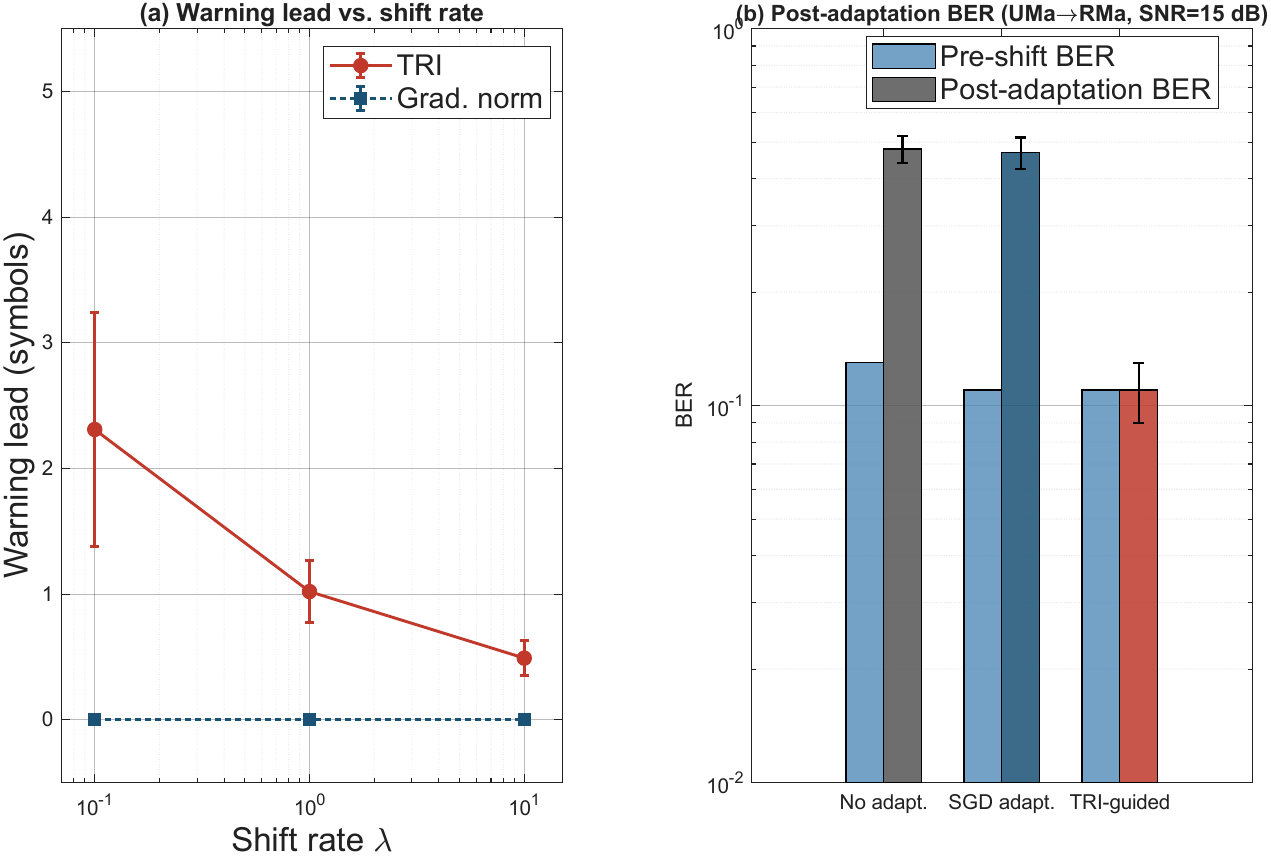}
  \vspace{-7mm}\caption{(a)~Warning lead vs.\ shift rate $\lambda$ for TRI and gradient
    norm (mean\,$\pm$\,std, 100 trials, SNR\,$=$\,15\,dB).
    (b)~BER vs.\ SNR for UMa$\to$RMa; TRI-guided adaptation achieves
    2.1\,dB gain over standard SGD at BER\,$=10^{-3}$.}
  \label{fig:metrics}
\vspace{-1mm}\end{figure}

\vspace{-0mm}\subsection{Computational Overhead}\vspace{-2mm}

TRI is evaluated every $T_{\rm eval}=50$ symbols using a Kalman-filtered
estimate of $\Phi_{\rm CM}$ (gain $K=0.3$), adding negligible overhead
to the inference pipeline. The Kalman filter reduces $\Phi_{\rm CM}$
variance and smooths transient fluctuations, contributing to the
consistent advance warning observed across all 10 scenarios.

\vspace{-5mm}\section{Conclusion}
\label{sec:conclusion}\vspace{-2mm}

We presented TRI, a principled metric for quantifying the resilience of AI-native wireless RXs under distributional channel shifts. By combining persistence exponents from loss-landscape, parameter-trajectory, and channel-manifold geometry, TRI provides a mean warning lead of $95$ OFDM symbols over BER-based detection. We establish boundedness, monotonicity, and Lipschitz stability, with Corollary~1 linking TRI to adaptation excess risk. TRI-guided adaptation achieves a $2.1$,dB SNR gain at BER,$=10^{-3}$ over standard SGD.
Beyond evaluation, TRI can serve as a topology-aware regularizer during training, a certification metric for safety-critical deployment, and a natural extension to end-to-end AI-native transceivers with jointly learned encoders and decoders, where richer coupled topological structure arises. Our framework opens a principled path toward topology-aware training and resilience monitoring for next-generation AI-native air interfaces.

\vspace{-2mm}\bibliographystyle{IEEEbib}
\def\baselinestretch{0.8}
\bibliography{references}

@techreport{itu2030,
  title={{IMT-2030 framework recommendation}},
  author={{ITU-R}},
  institution={International Telecommunication Union Radiocommunication Sector (ITU-R)},
  number={ITU-R M.2160},
  month={nov},
  year={2023}
}

@article{oshea2017,
  title={An introduction to deep learning for the physical layer},
  author={T. O'Shea and J. Hoydis},
  journal={IEEE Transactions on Cognitive Communications and Networking},
  volume={3},
  number={4},
  pages={563--575},
  month={dec},
  year={2017},
  publisher={IEEE}
}

@inproceedings{li2018landscape,
  title={Visualizing the loss landscape of neural nets},
  author={H. Li and Z. Xu and G. Taylor and C. Studer and T. Goldstein},
  booktitle={Advances in Neural Information Processing Systems (NeurIPS)},
  pages={6389--6399},
  address={Montreal, Canada},
  month={dec},
  year={2018}
}

@article{dorner2018,
  title={Deep learning based communication over the air},
  author={S. D{\"o}rner and S. Cammerer and J. Hoydis and S. ten Brink},
  journal={IEEE Journal of Selected Topics in Signal Processing},
  volume={12},
  number={1},
  pages={132--143},
  month={feb},
  year={2018},
  publisher={IEEE}
}

@book{edelsbrunner2010,
  title={Computational Topology: An Introduction},
  author={H. Edelsbrunner and J. Harer},
  publisher={American Mathematical Society},
  address={Providence, RI},
  year={2010}
}

@techreport{3gpp_tr38901,
  author      = {3GPP},
  title       = {Study on Channel Model for Frequencies from 0.5 to 100 {GHz}},
  institution = {3rd Generation Partnership Project},
  number      = {TR 38.901},
  type        = {Technical Report},
  edition     = {17.0.0},
  month       = jun,
  year        = {2022},
}

@techreport{series2017imt2020,
  title={Guidelines for evaluation of radio interface technologies for IMT-2020},
  author={M. Series},
  institution={International Telecommunication Union (ITU)},
  number={Report ITU-R M.2512},
  year={2017}
}

@article{barbarossa2020topological,
  title={Topological signal processing over simplicial complexes},
  author={S. Barbarossa and S. Sardellitti},
  journal={IEEE Transactions on Signal Processing},
  volume={68},
  pages={2992--3007},
  year={2020}
}

@article{BenDavidEtAl2010DomainAdaptation,
  title   = {A Theory of Learning from Different Domains},
  author  = {S. Ben-David and J. Blitzer and K. Crammer and A. Kulesza and F. Pereira and J.~W. Vaughan},
  journal = {Machine Learning},
  year    = {2010},
  volume  = {79},
  number  = {1},
  pages   = {151--175}
}

@book{khalil2002nonlinear,
  author    = {H. K. Khalil},
  title     = {Nonlinear Systems},
  publisher = {Prentice Hall},
  edition   = {3rd},
  year      = {2002}
}

@article{YangEtAl2020DLTransferFDD,
  title   = {Deep Transfer Learning-Based Downlink Channel Prediction for {FDD} Massive {MIMO} Systems},
  author  = {Y. Yang and F. Gao and Z. Zhong and B. Ai and A. Alkhateeb},
  journal = {IEEE Transactions on Communications},
  year    = {2020},
  volume  = {68},
  number  = {12},
  pages   = {7485--7497}
}

@article{RavivEtAl2023OnlineMetaLearning,
  title   = {Online Meta-Learning for Hybrid Model-Based Deep Receivers},
  author  = {T. Raviv and S. Park and O. Simeone and Y.~C. Eldar and N. Shlezinger},
  journal = {IEEE Transactions on Wireless Communications},
  year    = {2023},
  volume  = {22},
  number  = {10},
  pages   = {6415--6431}
}

@article{matthiesen2025resilient,
  title={Resilient Radio Access Networks: AI and the Unknown Unknowns},
  author={B. Matthiesen and A. Dekorsy and P. Popovski},
  journal={arXiv preprint arXiv:2510.21587},
  year={2025}
}

@incollection{edelsbrunner2008persistent,
  title={Persistent homology—A survey},
  author={H. Edelsbrunner and J. Harer},
  booktitle={Contemporary Mathematics},
  volume={453},
  pages={257--282},
  year={2008},
  publisher={American Mathematical Society}
}

@article{kalinke2022mmd,
  title={Maximum mean discrepancy on exponential windows for online change detection},
  author={F. Kalinke and M. Heyden and G. Gntuni and E. Fouch{\'e} and K. B{\"o}hm},
  journal={arXiv preprint arXiv:2205.12706},
  year={2022}
}

@article{gong2022cusum,
  title={Neural network-based CUSUM for online change-point detection},
  author={T. Gong and J. Lee and X. Cheng and Y. Xie},
  journal={arXiv preprint arXiv:2210.17312},
  year={2022}
}

@inproceedings{huang2021importance,
  title={On the importance of gradients for detecting distributional shifts in the wild},
  author={R. Huang and A. Geng and Y. Li},
  booktitle={Advances in Neural Information Processing Systems},
  volume={34},
  pages={677--689},
  year={2021}
}

@article{liu2026universal,
  title={A Universal Neural Receiver that Learns at the Speed of Wireless},
  author={L. Liu and L. Zheng and Y. Yi and R. Calderbank},
  journal={arXiv preprint arXiv:2602.15458},
  year={2026}
}

@article{yu2023adaptive,
  title={An adaptive and robust deep learning framework for THz ultra-massive MIMO channel estimation},
  author={W. Yu and Y. Shen and H. He and X. Yu and S. Song and J. Zhang and K.B. Letaief},
  journal={IEEE Journal of Selected Topics in Signal Processing},
  volume={17},
  number={4},
  pages={761--776},
  year={2023}
}

@inproceedings{shui2024resilient,
  title={Design and analysis of resilient vehicular platoon systems over wireless networks},
  author={T. Shui and W. Saad},
  booktitle={GLOBECOM 2024--2024 IEEE Global Communications Conference},
  pages={5186--5192},
  year={2024},
  month={December}
}

@article{mahmood2025resilient,
  title={Resilient-By-Design: A Resilience Framework for Future Wireless Networks},
  author={N.H. Mahmood and S. Samarakoon and P. Porambage and M. Bennis and M. Latva-Aho},
  journal={IEEE Communications Magazine},
  year={2025}
}

@article{reifert2023comeback,
  title={Comeback kid: Resilience for mixed-critical wireless network resource management},
  author={R.J. Reifert and S. Roth and A.A. Ahmad and A. Sezgin},
  journal={IEEE Transactions on Vehicular Technology},
  volume={72},
  number={12},
  pages={16177--16194},
  year={2023}
}

@article{sionna,
  title={Sionna: An open-source library for next-generation physical layer research},
  author={J. Hoydis and F. A. Aoudia and A. Valcarce and H. Viswanathan},
  journal={arXiv preprint arXiv:2203.11854},
  year={2022}
}

@book{chazal2016,
  title={The Structure and Stability of Persistence Modules},
  author={F. Chazal and V. de Silva and M. Glisse and S. Oudot},
  publisher={Springer},
  address={Cham},
  year={2016}
}

@article{ollivier2009,
  title={Ricci curvature of Markov chains on metric spaces},
  author={Y. Ollivier},
  journal={Journal of Functional Analysis},
  volume={256},
  number={3},
  pages={810--864},
  year={2009}
}

@article{hazan2016,
  title={Introduction to online convex optimization},
  author={E. Hazan},
  journal={Foundations and Trends in Optimization},
  volume={2},
  number={3--4},
  pages={157--325},
  year={2016}
}

@misc{gudhi,
  title={GUDHI User and Reference Manual},
  author={{The GUDHI Project}},
  howpublished={\url{https://gudhi.inria.fr/}},
  year={2024}
}
\end{document}